\documentclass[11pt,a4]{article}
\usepackage{amsfonts}
\usepackage{graphics,amssymb,amsmath,epsf,amsthm}
\usepackage{graphicx}
\usepackage{subfigure,color}

\newtheorem{theorem}{Theorem}[section]
\newtheorem{lemma}{Lemma}[section]

\newtheorem{example}{Example}[section]

\allowdisplaybreaks

\begin{document}

\title{Characterization of Order Statistics in Two Runs Using Conditional Expectation}

\author{M. R. Kazemi$^{1}$, A. A. jafari$^{2}$\thanks{Corresponding author. A. A. Jafari. Email: aajafari@yazd.ac.ir}
 \\
{\small $^1$Department of Statistics, Fasa University, Fasa, Iran}\\
{\small $^{2}$Department of Statistics, Yazd University, Yazd, Iran} }
\date{}
\maketitle

\begin{abstract}
The runs test is a well-known test that is used for checking independence between elements of a sample data sequence. Some of runs tests are based on the longest run and others based on the total runs. In this paper, we consider order statistics of two runs statistics, and obtain their probability mass functions. In addition, the means and variances of the order statistics are derived using traditional conditional expectation.
\end{abstract}

\section{Introduction}

One can use a test based upon the notion of runs to test of the hypothesis $H_0:F(z)=G(z)$, for all $z$
(see for example Hogg and Craig \cite{ho-cr-78}, pages 322-326),
 where $F$ and $G$ are two univariate discrete distribution functions. This notation is defined as follows. Let $x_{1},x_{2},...,x_{n_1}$ and $y_{1},y_{2},...,y_{n_2}$ be two
random samples of sizes $n_{1}$ and $n_{2}$ from $F$ and $G$, respectively. By combining these two samples, we have a new sample of size $n=n_{1}+n_{2}$. These values are arranged in order from smallest to largest, and a run constitutes when one or more values of $x$ or $y$ occur together.

The runs test is one of the easiest tests for checking the randomness hypothesis for a two-valued data sequence. More precisely, it can be used to test the hypothesis that the
elements of the sequence are mutually independent. The theory of runs
has been studied in several books such as
Gibbons and Chakraborty \cite{gi-ch-03},
Hogg and Craig \cite{ho-cr-78},
Randles and Wolfe \cite{ra-wo-79}
 and
Govindarajulu \cite{govindarajulu-07},
 and has various applications, for example in reliability, quality control, and so on. Some
runs' tests are based on the longest run length, while others are based on the
number of total runs. Mood \cite{mood-40}
derived the distribution of runs of given length of fixed number of elements of
two or more kinds. \cite{wa-wo-40} 
used runs statistic to test whether two samples are from the same
population, and  Wolfowitz \cite{wolfowitz-44} 
derived an asymptotic distribution of the nonparametric runs statistic.

Asano \cite{asano-65} 
found the distribution of runs statistic and suggested a method for testing
whether two samples observed on a circle are drawn from the same
distribution. Schwager \cite{schwager-83}  
considered $n$ trials with $\nu \geq 2$ outcomes, and computed the
probability occurrence of a given success run, success-failure runs, and
multiple runs in sequences of Markov dependent trials. Godbole and  Gomowiez \cite{go-go-92}
found the exact and approximate distributions of the total number of runs in
sequences of $n$ Bernoulli trials.
Lou \cite{lou-96}
studied the the distribution of runs using a method of finite Markov chain
imbedding, and found the exact joint and conditional distribution of runs.

To analyze runs and find their distributions, some methods use the
traditional combinatorial approach, and some use properties based in Markov
chain. For a complete review on these approaches, refer to Balakrishnan and Koutras \cite{ba-ko-02}.
An efficient method to derive the mean and variance of the total
number of runs based on a conditional approach was utilized Marshall \cite{marshall-70}. 



In this paper, we consider the order statistics of two runs of $x$ values and of $y$ values, i.e. maximum of two runs and minimum of two runs, and we
find their probability mass functions (pmf's), means, and variances using the conditional approach.
This paper is organised as follows: in Section \ref{sec.per}, some properties of runs are proposed. The
pms's of order statistics and their properties are given in Section \ref{sec.ord}. For illustration purpose, two examples are studied in Section \ref{sec.exm}.

\section{Preliminary study}
\label{sec.per}
Let $x_{1},x_{2},...,x_{n_{1}}$ and $y_{1},y_{2},...,y_{n_{2}}$ be two
random samples of sizes $n_{1}$ and $n_{2}$, respectively selected from a same population. The total number of runs, total number of runs of $x$'s,
total number of runs of $y$'s are denoted by $R$, $R_{1}$ and $R_{2}$,
respectively. The following two properties are easy to understand (see for example Hogg and Craig 
\cite{ho-cr-78}, page 324):

\noindent i) If the arrangement begins and finishes with $x$ values (or $y$ values), then $R_{1}=R_{2}\pm 1$.

\noindent ii) If the arrangement begins with $x$ values and finishes with $y$ values and vice versa, then $R_{1}=R_{2}$.

The $n_1$ objects of $x$'s values can be put into $r_1$ cells in $\binom{n_1-1}{r_1-1}$ ways and for each choice of this, $n_2$ objects of $y$'s values can be partitioned into $r_2$ parts in $\binom{n_2-1}{r_2-1}$ ways. If $r_1=r_2\pm 1$, a run of $y$'s or $x$'s values must come first. If $r_1=r_2$, the sequence can begin with a run of either type, so that the number of distinct arrangements have to be multiplied by 2. Also, we know that the total number of permutations of $n_1$objectsof $x$'s and $n_2$ objects of $y$'s values is
$\binom{n_1+n_2}{n_1}$. Therefore, the joint pmf of
nonparametric random variables $R_{1}$ and $R_{2}$ is
\begin{eqnarray*}
f_{R_{1},R_{2}}\left( r_{1},r_{2}\right) = \left\{
\begin{array}{ll}
v & \text{if }r_{1}=r_{2}\pm 1 \\
2v & \text{if }r_{1}=r_{2},%
\end{array}
\right.
\end{eqnarray*}
where $v=\frac{\binom{n_1-1}{r_1-1} \binom{n_2-1}{r_2-1}  }{\binom{n}{n_1} },$
and $n=n_1+n_2$.
Readers can see for example
Hogg and Craig
\cite{ho-cr-78} page 78, Govindarajulu \cite{govindarajulu-07} page 206, and Gibbons and Chakraborty \cite{gi-ch-03} page 79.

Using the following lemma, we can compute the $P\left( R_{1}>R_{2}\right)
$, $P\left( R_{1}<R_{2}\right)$ and $P\left( R_{1}=R_{2}\right)$.

\begin{lemma}
\label{lemma.P} Let $R_1$ and $R_2$ be the total number of runs of $x$ values and $y$ values, respectively. Then

\noindent a. $$P\left(R_1=R_2\right)=\dfrac{2n_1n_2}{n(n-1)}$$

\noindent b. $$P\left(R_1>R_2\right)=\dfrac{n_1(n_1-1)}{n\left(n-1\right)}$$

\noindent c. $$P\left(R_1<R_2\right)=\dfrac{n_2(n_2-1)}{n\left(n-1\right)}.$$
\end{lemma}

\begin{proof}
\noindent a.
\begin{eqnarray*}
P\left(R_1=R_2\right)&=&\sum^{\min(n_1,n_2)}_{i=1}{P\left(R_1=i,R_2=i\right)}\\
&=&\frac{2}{\left( \begin{array}{c}
n_1+n_2 \\
n_1 \end{array}
\right)}\sum^{\min(n_1,n_2)}_{i=1}{\left( \begin{array}{c}
n_1-1 \\
i-1 \end{array}
\right)\left( \begin{array}{c}
n_2-1 \\
i-1 \end{array}
\right)}\\
&=&\frac{2}{\left( \begin{array}{c}
n_1+n_2 \\
n_1 \end{array}
\right)}\frac{\Gamma \left(n_1+n_2-1\right)}{\Gamma \left(n_1\right)\Gamma \left(n_2\right)}\\
&=&\frac{2n_1n_2}{n\left(n-1\right)}
\end{eqnarray*}

\noindent b.
\begin{eqnarray*}
P\left(R_1>R_2\right)&=&\sum^{n_1-1}_{i=1}{P\left(R_1=i+1,R_2=i\right)}\\
&=&\frac{1}{\left( \begin{array}{c}
n_1+n_2 \\
n_1 \end{array}
\right)}\sum^{n_1-1}_{i=1}{\left( \begin{array}{c}
n_1-1 \\
i \end{array}
\right)\left( \begin{array}{c}
n_2-1 \\
i-1 \end{array}
\right)}\\
&=&\frac{1}{\left( \begin{array}{c}
n_1+n_2 \\
n_1 \end{array}
\right)}\frac{\Gamma \left(n_1+n_2-1\right)}{\Gamma \left(n_1-1\right)\Gamma \left(n_2+1\right)}\\
&=&\frac{n_1(n_1-1)}{n\left(n-1\right)}.
\end{eqnarray*}

\noindent c. The proof of part c is  the same as part b.
\end{proof}

\section{The order statistics of runs}
\label{sec.ord}
In this section, we characterize the nonparametric statistics $R_{M}=\max\left(R_{1},R_{2}\right) $ and $R_{m}=\min\left(R_{1},R_{2}\right) $ by means of conditional expectation, and we derive their expectation values, variances, and covariance between them. It can be shown that the pmf of the nonparametric statistic $R_{M}$ has the following form:
\begin{eqnarray*}
P\left(R_M=t\right)&=&
P\left(R_M=t|R_1>R_2\right)P\left(R_1>R_2\right)\hspace{2.5cm}\\
&&+P\left(R_M=t|R_1<R_2\right)P\left(R_1<R_2\right)\\
&&+P\left(R_M=t|R_1=R_2\right)P\left(R_1=R_2\right)\\
&=&P\left(R_1=t,R_2=t-1\right)P\left(R_1>R_2\right)\\
&&+P\left(R_1=t-1,R_2=t\right)P\left(R_1<R_2\right)\\
&&+P\left(R_1=t,R_2=t\right)P\left(R_1=R_2\right)\\
&=&\frac{\left( \begin{array}{c}
n_1-1 \\
t-1 \end{array}
\right)\left( \begin{array}{c}
n_2-1 \\
t-2 \end{array}
\right)}{\left( \begin{array}{c}
n_1+n_2 \\
n_1 \end{array}
\right)}.\frac{n_1\left(n_1-1\right)}{n\left(n-1\right)}
\\
&&+\frac{\left( \begin{array}{c}
n_1-1 \\
t-2 \end{array}
\right)\left( \begin{array}{c}
n_2-1 \\
t-1 \end{array}
\right)}{\left( \begin{array}{c}
n_1+n_2 \\
n_1 \end{array}
\right)}.\frac{n_2\left(n_2-1\right)}{n\left(n-1\right)}\\
&&+\frac{\left( \begin{array}{c}
n_1-1 \\
t-1 \end{array}
\right)\left( \begin{array}{c}
n_2-1 \\
t-1 \end{array}
\right)}{\left( \begin{array}{c}
n_1+n_2 \\
n_1 \end{array}
\right)}.\frac{2n_1n_2}{n\left(n-1\right)},  \qquad t=1,2,\dots ,\max(n_1,n_2).
\end{eqnarray*}

Computing the mean and variance of $R_{M}$ by using the above
pmf requires tedious calculations. To
solve this problem, we make use of another method which is used by Marshall \cite{marshall-70}. 
We claim that our method is simple and has pedagogical advantages.

As known, the sample space $\Omega $ can be partitioned as $\Omega =\left\{
R_{1}<R_{2}\right\} \cup \left\{ R_{1}>R_{2}\right\} \cup \left\{
R_{1}=R_{2}\right\}.$ Therefore,
\begin{eqnarray}  \label{eq.ERM}
E\left( R_{M}\right)&=& E\left( R_{M}|R_{1}>R_{2}\right) P\left(R_{1}>R_{2}%
\right)  \notag \\
&&+E\left( R_{M}|R_{1}<R_{2}\right) P\left( R_{1}<R_{2}\right) \\
&&+E\left( R_{M}|R_{1}=R_{2}\right) P\left( R_{1}=R_{2}\right).  \notag
\end{eqnarray}

To simplify (\ref{eq.ERM}), we need to compute $E\left(R_{M}|R_{1}>R_{2}%
\right)$, $E\left( R_{M}|R_{1}<R_{2}\right)$, and \linebreak
 $E\left( R_{M}|R_{1}=R_{2}%
\right)$. We can compute these values using the following lemma.

\begin{lemma}
\label{lemma.E} Let $R_1$ and $R_2$ be the total number of runs of $x$ values and $y$ values, respectively. Then

\noindent i. $$E\left(R_M\left|R_1>R_2\right.\right)=2+\frac{\left(n_1-2\right)(n_2-1)}{n-2},$$

\noindent ii. $$E\left(R_M\left|R_1<R_2\right.\right)=2+\frac{(n_1-1)\left(n_2-2\right)}{n-2},$$

\noindent iii. $$E\left(R_M\left|R_1=R_2\right.\right)=1+
\frac{\left(n_1-1\right)(n_2-1)}{n-2},$$\\
where $n=n_1+n_2$.

\end{lemma}

\begin{proof}
\noindent i. Note that when $R_{1}>R_{2}$,
the arrangement begins with some of $x$ observations and finishes with other $x$ observations.
Therefore, we can use the following diagram:

Allow  $y$ values to be laid out in a row, and place $x$ values in the spaces between those $y$ values, and let our diagram begin
and finish with some of $x$ values. Then, the total number of runs of
$x$ values is distributed as follows
\begin{equation}
\overset{\overset{W_{1}}{\uparrow }}{-},y_{1},\overset{\overset{W_{2}}{%
\uparrow }}{-},y_{2},\overset{\overset{W_{3}}{\uparrow }}{-}%
,y_{3},...,y_{n_2-1},\overset{\overset{W_{n_2}}{\uparrow }}{-},y_{n_2},\overset{%
\overset{Wn_2+1}{\uparrow }}{-}\tag{D-1}
\end{equation}
where $W_{i},i=2,3,...,n_2$ has the following form
\begin{equation*}
W_{i}=\left\{
\begin{array}{l}
1, \\
0,%
\end{array}
\begin{array}{l}
\text{a run occurs in place}\ i \\
\text{other wise,}%
\end{array}\right.
\end{equation*}%
and $W_{1}$ and $W_{n_2+1}$ are equal to 1. Using the above discussion, we
can write
\begin{equation*}
\left( R_{1}|R_{1}>R_{2}\right) =2+\sum_{i=2}^{n_2}W_{i}.
\end{equation*}%
 Therefore,
\begin{eqnarray*}
E\left( R_{M}|R_{1}>R_{2}\right) =2+\sum_{i=2}^{n_2}E\left( W_{i}\right)
=2+\sum_{i=2}^{n_2}P\left( \text{a run occurs in place }i\right) .
\end{eqnarray*}

The event $\left\{\text{a run occurs in place }i\right\} $ is equivalent to
the case that we have $n_2+1$ boxes and we want to put $n_1-2$ nuts into these $n_2+1$ boxes where the box $i$
 should consist of
at least one $x$ (nut). Then,
\begin{equation*}
P\left( \text{a run occurs in place} \ i\right) =\frac{\left(
\begin{array}{c}
n-3 \\
n_1-3%
\end{array}
\right) }{\left(
\begin{array}{c}
n-2 \\
n_1-2%
\end{array}
\right) }=\frac{n_1-2}{n-2}.
\end{equation*}%
So we have
\begin{equation*}
E\left( R_{M}|R_{1}>R_{2}\right) =2+\frac{\left( n_1-2\right) \left(
n_2-1\right) }{\left( n-2\right) }.
\end{equation*}

\noindent ii. The proof of this part is the same as part i.

\noindent iii. The proof of this part also is the same as
part (i) with only some changes in the diagram. When the event $\left\{R_{1}=R_{2}\right\} $
occurs, it means that in the permutation of $x$ values and $y$ values, permutation begins with some $x$ values and ends with
some $y$ values and vice versa as below:
\begin{equation}
\overset{\overset{W_{1}}{\uparrow }}{-},y_{1},\overset{\overset{W_{2}}{%
\uparrow }}{-},y_{2},\overset{\overset{W_{3}}{\uparrow }}{-}%
,y_{3},...,y_{n_2-1},\overset{\overset{W_{n_2}}{\uparrow }}{-},y_{n_2}\ \ \ \ \ \
\ \text{or \ \ \ \ \ \ \ \ }\overset{\overset{W_{1}^*}{\uparrow }}{-},x_{1},%
\overset{\overset{W_{2}^*}{\uparrow }}{-},x_{2},\overset{\overset{W_{3}^*}{%
\uparrow }}{-},x_{3},...,x_{n_1-1},\overset{\overset{W_{n_1}^*}{\uparrow }}{-}%
,x_{n_1}\tag{D-2}
\end{equation}%
where $W_{i},i=2,3,...,n_2$ and $W_{j}^*,j=2,3,...,n_1$ have the following form
\bigskip
\begin{equation*}
W_{i}=\left\{%
\begin{array}{l}
1, \\
0,%
\end{array}%
\begin{array}{l}
\text{a run occurs in place}\ i \\
\text{other wise,}%
\end{array}\right.
\qquad
W_{j}^*=\left\{%
\begin{array}{l}
1, \\
0,%
\end{array}%
\begin{array}{l}
\text{a run occurs in place}\ j \\
\text{other wise,}%
\end{array}\right.
\end{equation*}%
and $W_{1}$ and $W_{1}^*$ are equal $1$. By using the above discussion
and since
\[R_M=\frac{1}{2}\left\{R_1+R_2+\left|R_1-R_2\right|\right\},\]
we have
\begin{equation*}
\left( R_{M}|R_{1}=R_{2}\right) =1+\frac{1}{2}\sum_{i=2}^{n_2}W_{i}+\frac{1}{2}\sum_{i=2}^{n_1}W_{i}^*=1+\sum_{i=2}^{n_2}W_{i},
\end{equation*}%
Therefore,
\begin{equation*}
E\left( R_{M}|R_{1}=R_{2}\right) =
1+\sum_{i=2}^{n_2}E(W_{i})
=1+\frac{\left(n_1-1\right)\left(n_2-1\right)}{n-2},
\end{equation*}%
and the proof is completed.
\end{proof}

\begin{lemma}
\label{lemma.V} Let $R_1$ and $R_2$ be the total numbers of runs of $x$ values and $y$ values, respectively. Then,

\noindent i.
\[Var\left(R_M\mathrel{\left|\vphantom{R_M R_1>R_2}\right.\kern-\nulldelimiterspace}R_1>R_2\right)=\frac{n_2\left(n_2-1\right)\left(n_1-2\right)\left(n_1-1\right)}{{\left(n-2\right)}^2\left(n-3\right)},\ \ \ n>3,\]
\noindent ii.
\[Var\left(R_M\mathrel{\left|\vphantom{R_M R_1<R_2}\right.\kern-\nulldelimiterspace}R_1<R_2\right)=\frac{n_1\left(n_1-1\right)\left(n_2-2\right)\left(n_2-1\right)}{{\left(n-2\right)}^2\left(n-3\right)},\ \ \ n>3,\]
\noindent iii.
\[Var\left(R_M\mathrel{\left|\vphantom{R_M R_1=R_2}\right.\kern-\nulldelimiterspace}R_1=R_2\right)=\frac{{\left(n_2-1\right)}^2{\left(n_1-1\right)}^2}{{\left(n-2\right)}^2\left(n-3\right)},\ \ \ n>3.\]
\end{lemma}

\begin{proof}

i. From (D-1) we have
\begin{eqnarray*}
Var\left(W_i\right)&=&E\left(W^2_i\right)-{\left(E\left(W_i\right)\right)}^2
=E\left(W_i\right)-{\left(E\left(W_i\right)\right)}^2\\
&=&E\left(W_i\right)\left[1-E\left(W_i\right)\right]
=\frac{n_1-2}{n-2}\left(1-\frac{n_1-2}{n-2}\right)
=\frac{\left(n_1-2\right)n_2}{{\left(n-2\right)}^2}.
\end{eqnarray*}
Also, we need to verify $E(W_i W_j )$. For this, note that
\begin{eqnarray*}
E\left( W_{i}W_{j}\right) =P\left( \text{two runs occur in places }i\text{
and }j\right).
\end{eqnarray*}
The event $\left\{ \text{two runs occur in places }i\text{ and }j\right\} $ is equivalent to the case that we have $n_2+1$ boxes and we want to put $n_1-2$ nuts into these $n_2+1$ boxes where the boxes $i$ and $j$ should consist of at least one $x$. Then
\begin{eqnarray*}
P\left( \text{two runs occur in places }i\text{ and }j\right) =P\left(
W_{i}=1\text{ and }W_{j}=1\right) =\frac{\left(
\begin{array}{c}
n-4 \\
n_{1}-4%
\end{array}%
\right) }{\left(
\begin{array}{c}
n-2 \\
n_{1}-2%
\end{array}%
\right) }=\frac{(n_{1}-2)(n_{1}-3)}{(n-2)(n-3)}.
\end{eqnarray*}
And then, for $i\ne j$
\begin{eqnarray*}
Cov\left(W_i,W_j\right)&=&E\left(W_iW_j\right)-E\left(W_i\right)E\left(W_j\right)\\
&=&\frac{\left(n_1-2\right)\left(n_1-3\right)}{\left(n-2\right)\left(n-3\right)}-\frac{{\left(n_1-2\right)}^2}{{\left(n-2\right)}^2}
=-\frac{\left(n_1-2\right)n_2}{{\left(n-2\right)}^2\left(n-3\right)}.
\end{eqnarray*}
Therefore,
\begin{eqnarray*}
Var\left(R_M\mathrel{\left|\vphantom{R_M R_1>R_2}\right.\kern-\nulldelimiterspace}R_1>R_2\right)
&=&Var\left(2+\sum^{n_2}_{i=2}{W_i}\right)=\sum^{n_2}_{i=2}{Var\left(W_i\right)}+\sum^{n_2}_{i\ne j}{Cov(W_i,W_j)}\\
&=&\frac{\left(n_2-1\right)\left(n_1-2\right)n_2}{{\left(n-2\right)}^2}-\frac{\left(n_2-1\right)
\left(n_2-2\right)\left(n_1-2\right)n_2}{{\left(n-2\right)}^2\left(n-3\right)}\\
&=&\frac{n_2\left(n_2-1\right)\left(n_1-2\right)\left(n_1-1\right)}{{\left(n-2\right)}^2\left(n-3\right)}.
\end{eqnarray*}

\noindent ii. The proof is similar to part i.

\noindent iii. From (D-2) we have
\begin{eqnarray*}
Var\left(W_i\right)=E\left(W^2_i\right)-{\left(E\left({W}_i\right)\right)}^2
=E\left(W_i\right)-{\left(E\left(W_i\right)\right)}^2\\
=E\left(W_i\right)\left[1-E\left(W_i\right)\right]
=\frac{\left(n_1-1\right)(n_2-1)}{{\left(n-2\right)}^2}.
\end{eqnarray*}
Also, for $i\ne j$
\begin{eqnarray*}
Cov\left(W_i,W_j\right)&=&E\left(W_iW_j\right)-E\left(W_i\right)E\left(W_j\right)
=\frac{\left(n_1-1\right)\left(n_1-2\right)}{\left(n-2\right)\left(n-3\right)}
-\frac{{\left(n_1-1\right)}^2}{{\left(n-2\right)}^2}\\
&=&\frac{n_1-1}{n-2}\left(\frac{n_1-2}{n-3}-\frac{n_1-1}{n-2}\right)
=-\frac{\left(n_1-1\right)\left(n_2-1\right)}{\left(n-3\right){\left(n-2\right)}^2}.
\end{eqnarray*}
Therefore
\begin{eqnarray*}
Var\left(R_M\mathrel{\left|\vphantom{R_M R_1=R_2}\right.\kern-\nulldelimiterspace}R_1=R_2\right)&=&Var\left(R_1\mathrel{\left|\vphantom{R_1 R_1=R_2}\right.\kern-\nulldelimiterspace}R_1=R_2\right)=Var\left(1+\sum^{n_2}_{i=2}{W_i}\right)\\
 &=&\sum^{n_2}_{i=2}{Var\left(W_i\right)}+\sum^{n_2}_{i\ne j}{Cov(W_i,W_j)}\\
 &=&\frac{{(n_2-1)}^2\left(n_1-1\right)}{{\left(n-2\right)}^2}
 -\frac{{\left(n_2-1\right)}^2\left(n_1-1\right)\left(n_2-2\right)}{\left(n-3\right){\left(n-2\right)}^2}\\
 &=&
 \frac{{\left(n_2-1\right)}^2{\left(n_1-1\right)}^2}{{\left(n-2\right)}^2\left(n-3\right)}.
\end{eqnarray*}
\end{proof}

\begin{theorem}\label{thm.EVM}
Let $R_{1}$ and $R_{2}$ be two runs of $x$ values and $y$ values, respectively. Then,
\begin{eqnarray}
E(R_M)=2+\frac{n(n_1-1)(n_2-1)-2n_1n_2}{n(n-1)}\hspace{4cm}\\
Var(R_M)=\frac{n_1n_2[n_1^3+n_2^3+n_1^3n_2+n_1n_2^3-n_1^2-n_2^2+2n_1
^2n_2^2-5n_1^2n_2-5n_1n_2^2+6n_1n_2]}{n^2(n-1)^2(n-2)}.
\end{eqnarray}
\end{theorem}

\begin{proof} See the Appendix A.1.
\end{proof}

\begin{lemma}\label{lem.Rm1}
Let $R_{{\rm 1}}$ and $R_{{\rm 2}}$ be the total number of runs of $x$ values and $y$ values, respectively. Then

\noindent i.
\[E\left(R_m\left|R_{{\rm 1}}{\rm >}R_{{\rm 2}}\right.\right){\rm =2+}\frac{\left(n_{{\rm 1}}{\rm -}{\rm 1}\right){\rm (}n_{{\rm 2}}{\rm -}{\rm 2)}}{n{\rm -}{\rm 2}},\]
\noindent ii.
\[{\rm E}\left(R_m\left|R_{{\rm 1}}{\rm <}R_{{\rm 2}}\right.\right){\rm =2+}\frac{{\rm (}n_{{\rm 1}}{\rm -}{\rm 2)}\left(n_{{\rm 2}}{\rm -}{\rm 1}\right)}{n{\rm -}{\rm 2}},\]
\noindent iii.
\[E\left(R_m\left|R_{{\rm 1}}{\rm =}R_{{\rm 2}}\right.\right){\rm =1+}\frac{\left(n_{{\rm 1}}{\rm -}{\rm 1}\right){\rm (}n_{{\rm 2}}{\rm -}{\rm 1)}}{n{\rm -}{\rm 2}}.\]
\end{lemma}
\begin{proof}
The proof is similar to  Lemma \ref{lemma.E}.
\end{proof}

\begin{lemma}\label{lem.Rm2}
Let $R_{{\rm 1}}$ and $R_{{\rm 2}}$ be the total numbers of runs of $n_{{\rm 1}}x$-type values and $n_{{\rm 2}}y$-type values, respectively. Then,

\noindent i.
\[Var\left(R_m\left|R_{{\rm 1}}{\rm >}R_{{\rm 2}}\right.\right)
{\rm =}\frac{n_{{\rm 1}}\left(n_{{\rm 1}}{\rm -}{\rm 1}\right)\left(n_{{\rm 2}}{\rm -}{\rm 2}\right)\left(n_{{\rm 2}}{\rm -}{\rm 1}\right)}{{\left(n{\rm -}{\rm 2}\right)}^{{\rm 2}}\left(n{\rm -}{\rm 3}\right)}
{\rm ,\ \ \ \ \ \ }n{\rm >}3,\]
\noindent ii.
\[Var\left(R_m\left|R_{{\rm 1}}{\rm <}R_{{\rm 2}}\right.\right)
{\rm =}\frac{n_{{\rm 2}}\left(n_{{\rm 2}}{\rm -}{\rm 1}\right)\left(n_{{\rm 1}}{\rm -}{\rm 2}\right)\left(n_{{\rm 1}}{\rm -}{\rm 1}\right)}{{\left(n{\rm -}{\rm 2}\right)}^{{\rm 2}}\left(n{\rm -}{\rm 3}\right)}
{\rm ,\ \ \ \ \ \ }n{\rm >}3,\]
\noindent iii.
\[Var\left(R_m\left|R_{{\rm 1}}{\rm =}R_{{\rm 2}}\right.\right){\rm =}\frac{{\left(n_{{\rm 2}}{\rm -}{\rm 1}\right)}^{{\rm 2}}{\left(n_{{\rm 1}}{\rm -}{\rm 1}\right)}^{{\rm 2}}}{{\left(n{\rm -}{\rm 2}\right)}^{{\rm 2}}\left(n{\rm -}{\rm 3}\right)}{\rm ,\ \ \ \ \ \ }n{\rm >}3\].

\end{lemma}
\begin{proof}
The proof is similar to  Lemma \ref{lemma.V}.
\end{proof}

\begin{theorem}\label{thm.EVm}
Let $R_{1}$ and $R_{2}$ be two runs of $x$ values and $y$ values, respectively. Then
\begin{eqnarray*}
E(R_m)=\frac{n_1n_2}{n-1},\hspace{1.6cm}
Var(R_m)=\frac{n_1n_2(n_1-1)(n_2-1)}{{(n-1)}^2(n-2)}, \ \ \ \ n>2.
\end{eqnarray*}
\end{theorem}

\begin{proof}
See the Appendix A.2.
\end{proof}

\begin{theorem}\label{thm.cov}
Let $R_{1}$ and $R_{2}$ be two runs of $x$ values and $y$ values, respectively. Then,
\[Cov\left(R_m,R_M\right)=\frac{n_1n_2(n_1-1)(n_2-1)}{n{(n-1)}^2}.\]
\end{theorem}

\begin{proof}
The variance of the total number of runs, $R$, was obtained by
Marshall \cite{marshall-70}
which has the following form:
\[Var\left(R\right)=\frac{{2n}_1n_2\left({2n}_1n_2-n\right)}{n^2\left(n-1\right)}.\]
This formula helps us to find the covariance of the nonparametric statistics $R_m$ and $R_M$ as
\[Var\left(R\right)=Var\left(R_m+R_M\right)=Var\left(R_m\right)+Var\left(R_M\right)+2Cov\left(R_m,R_M\right),\]
and  the proof is completed.
\end{proof}

\section{Two illustrative examples}
\label{sec.exm}
In this section, we give two examples to illustrate how our calculations work.

\begin{example}
We take $n_1=3$ and $n_2=2$.  In the following table, we gave the possible permutations, number of runs of $x$ values ($R_1$),  number of runs of $y$ values ($R_2$), the minimum value of $R_1$ and $R_2$, and maximum value of $R_1$ and $R_2$.

\begin{center}
\begin{tabular}{|c|cccc|}\hline
permutations & $R_{1}$ & $R_{2}$ & $R_{m}$ & $R_{M}$ \\ \hline
$xxxyy$ & 1 & 1 & 1 & 1 \\
$xxyxy$ & 2 & 2 & 2 & 2 \\
$xxyyx$ & 2 & 1 & 1 & 2 \\
$xyxxy$ & 2 & 2 & 2 & 2 \\
$xyyxx$ & 2 & 1 & 1 & 2 \\
$xyxyx$ & 3 & 2 & 2 & 3 \\
$yxxxy$ & 1 & 2 & 1 & 2 \\
$yxxyx$ & 2 & 2 & 2 & 2 \\
$yxyxx$ & 2 & 2 & 2 & 2 \\
$yyxxx$ & 1 & 1 & 1 & 1 \\ \hline
\end{tabular}
\end{center}
\noindent  Therefore, The joint pmf of  $R_m$ and  $R_M$ has the following form:

\begin{center}
\begin{tabular}{c|ccc|c}
$\left( s,t\right) $ & 1 & 2 & 3 & $P\left( R_{m}=s\right) $ \\ \hline
1 & 0.2 & 0.3 & 0.0 & 0.5 \\
2 & 0.0 & 0.4 & 0.1 & 0.5 \\\hline
$P\left( R_{M}=t\right)$ & 0.2 & 0.7 & 0.1 & 1.0\\ 
\end{tabular}
\end{center}
Then, we have,
$$E(R_m)=1.5,  \quad  Var(R_m)=0.25, \quad   E(R_M)=1.9,  \quad  Var(R_M)=0.29,   \quad  cov(R_m,R_M)=0.15,$$
which agree with the results of Theorems \ref{thm.EVM}, \ref{thm.EVm}, and \ref{thm.cov}.
\end{example}

\begin{example}
Here, we considered different values of $n_1$ and $n_2$, and obtained the joint
pmf of $R_m$ and $R_M$.  The marginal pmf's of these statistics
are given in Table \ref{tab.1}. We also calculated the expectation, variance, and covariance of $R_m$ and $R_M$ directly using these pmf's. These are agree with the results of Theorems \ref{thm.EVM}, \ref{thm.EVm}, and \ref{thm.cov}.
\end{example}

\begin{table}[h]
\begin{center}
\caption{The marginal probability mass function, expectation, variance, and covariance of $R_m$ and $R_M$ for given $n_1$ and $n_2$}\label{tab.1}
\begin{tabular}{|c|c|c|c|c|c|c|c|c|c|c|} \hline
 & \multicolumn{10}{|c|}{$(n_1,n_2)$} \\ \hline
 & \multicolumn{2}{|c|}{3,3} & \multicolumn{2}{|c|}{12,3} & \multicolumn{2}{|c|}{10,5} & \multicolumn{2}{|c|}{8,7} & \multicolumn{2}{|c|}{9,9} \\ \hline
$i$ & $R_m$ & $R_M$ & $R_m$ & $R_M$ & $R_m$ & $R_M$ & $R_m$ & $R_M$ & $R_m$ & $R_M$ \\ \hline
1 & 0.3 & 0.1 & 0.033 & 0.004 & 0.005  & 0.001  & 0.002  & 0.000  & 0.000 & 0.000 \\
2 & 0.6 & 0.6 & 0.363 & 0.125 & 0.090  & 0.028 & 0.049 & 0.015 & 0.012 & 0.003 \\ \
3 & 0.1 & 0.3 & 0.604 & 0.508 & 0.360  & 0.210 & 0.245 & 0.134 & 0.097 & 0.041 \\ \
4 &  &  &  & 0.363 & 0.420  & 0.440 & 0.408 & 0.364 & 0.290 & 0.194 \\ \
5 &  &  &  &  & 0.126 & 0.280 & 0.245 & 0.354 & 0.363 & 0.363 \\
6 &  &  &  &  &  & 0.042 & 0.049 & 0.121 & 0.194 & 0.290 \\
7 &  &  &  &  &  &  & 0.002 & 0.012 & 0.041 & 0.097 \\
8 &  &  &  &  &  &  &  & 0.000 & 0.003 & 0.012 \\
9 &  &  &  &  &  &  &  &  & 0.000 & 0.000 \\ \hline
\multicolumn{1}{|l|}{Expectation} & 1.80 & 2.20 & 2.571 & 3.228 & 3.571 & 4.095 & 4.000 & 4.466 & 4.764 & 5.235 \\ \hline
\multicolumn{1}{|l|}{Variance}  & 0.36 & 0.36 & 0.310 & 0.453 & 0.706 & 0.767 & 0.923 & 0.925 & 1.121 & 1.121 \\ \hline
\multicolumn{1}{|l|}{Covariance} & \multicolumn{2}{|c|}{0.24} & \multicolumn{2}{|c|}{0.269} & \multicolumn{2}{|c|}{0.612} & \multicolumn{2}{|c|}{0.800} & \multicolumn{2}{|c|}{0.996} \\ \hline
\end{tabular}
\end{center}
\end{table}

\section{Conclusion}
\label{sec.appen}
The theory of runs has been studied in many literature and has various applications such as testing the identically distributed and randomness hypotheses. In this article, we consider the order statistics of runs, i.e. maximum and minimum of two runs, and find their pmf's. Computing the means and variances of these nonparametric statistics by using their pmf's require tedious calculations. Therefore, the conditional expectation method are used to solve this problem.

\section*{Acknowledgments}
The authors are thankful to the Editor and referees for helpful comments and suggestions.

\section*{Appendix}

\subsection*{A.1. Proof of Theorem \ref{thm.EVM} }
Using Lemmas \ref{lemma.P} and \ref{lemma.E}, we have
\begin{eqnarray*}
E\left(R_M\right)& =&E\left(R_M{\rm |}R_{{\rm 1}}{\rm >}R_{{\rm 2}}\right)P\left(R_{{\rm 1}}{\rm >}R_{{\rm 2}}\right){\rm +}E\left(R_M{\rm |}R_{{\rm 1}}{\rm <}R_{{\rm 2}}\right)P\left(R_{{\rm 1}}{\rm <}R_{{\rm 2}}\right)\\
&&+E\left(R_M{\rm |}R_{{\rm 1}}{\rm =}R_{{\rm 2}}\right)P\left(R_{{\rm 1}}{\rm =}R_{{\rm 2}}\right)\\
&=&\left({\rm 2+}\frac{\left(n_{{\rm 1}}{\rm -}{\rm 2}\right){\rm (}n_{{\rm 2}}{\rm -}{\rm 1)}}{n{\rm -}{\rm 2}}\right){\rm \times }\frac{n_1\left(n_1-1\right)}{n\left(n-1\right)}{\rm +}\left({\rm 2+}\frac{{\rm (}n_{{\rm 1}}{\rm -}{\rm 1)}\left(n_{{\rm 2}}{\rm -}{\rm 2}\right)}{n{\rm -}{\rm 2}}\right) \\
&& \times \frac{n_{{\rm 2}}\left(n_{{\rm 2}}{\rm -}{\rm 1}\right)}{n\left(n{\rm -}{\rm 1}\right)}{\rm +\ }\left({\rm 1+}\frac{\left(n_{{\rm 1}}{\rm -}{\rm 1}\right){\rm (}n_{{\rm 2}}{\rm -}{\rm 1)}}{n{\rm -}{\rm 2}}\right){\rm \times }\frac{{\rm 2}n_{{\rm 1}}n_{{\rm 2}}}{n\left(n{\rm -}{\rm 1}\right)}\\
&=&{\rm 2+}\frac{\left(n_{{\rm 1}}{\rm +}n_{{\rm 2}}\right)\left(n_{{\rm 1}}{\rm -}{\rm 1}\right)\left(n_{{\rm 2}}{\rm -}{\rm 1}\right){\rm -}{\rm 2}n_{{\rm 1}}n_{{\rm 2}}}{\left(n_{{\rm 1}}{\rm +}n_{{\rm 2}}\right)\left(n_{{\rm 1}}{\rm +}n_{{\rm 2}}{\rm -}{\rm 1}\right)}.
\end{eqnarray*}
Also,
\begin{eqnarray*}
Var\left( R_{M}\right) &=&Var\left( R_{M}|R_{1}>R_{2}\right) P\left(
R_{1}>R_{2}\right) +Var\left(R_{M}|R_{1}<R_{2}\right) P\left(
R_{1}<R_{2}\right) \\
&&+Var\left( R_{M}|R_{1}=R_{2}\right) P\left( R_{1}=R_{2}\right) +(E(R_{M}|R_{1}>R_{2}))^{2}
P\left(R_{1}>R_{2}\right) \\
&&+(E(R_{M}|R_{1}<R_{2}))^{2} P\left(R_{1}<R_{2}\right)+(E(R_{M}|R_{1}=R_{2}))^{2}
P\left( R_{1}=R_{2}\right) \\
&&-(E(R_{M}))^{2}.
\end{eqnarray*}
So, by substituting appropriate formulas, we have
\begin{eqnarray*}
Var(R_{M}) &=&\frac{n_{2}(n_{2}-1)(n_{1}-2)(n_{1}-1)}{(n-2)^{2}(n-3)}\times
\frac{n_{1}(n_{1}-1)}{n(n-1)} \\
&&+\frac{n_{1}(n_{1}-1)(n_{2}-2)(n_{2}-1)}{(n-2)^{2}(n-3)}\times \frac{%
n_{2}(n_{2}-1)}{n(n-1)} \\
&&+\frac{(n_{2}-1)^{2}(n_{1}-1)^{2}}{(n-2)^{2}(n-3)}\times \frac{2n_{1}n_{2}%
}{n(n-1)} \\
&&+\left( 2+\frac{(n_{1}-2)(n_{2}-1)}{n-2}\right) ^{2}\times \frac{%
n_{1}(n_{1}-1)}{n(n-1)} \\
&&+\left( 2+\frac{(n_{1}-1)(n_{2}-2)}{n-2}\right) ^{2}\times \frac{%
n_{2}(n_{2}-1)}{n(n-1)} \\
&&+\left( 1+\frac{(n_{1}-1)(n_{2}-1)}{n-2}\right) ^{2}\times \frac{%
2n_{1}n_{2}}{n(n-1)} \\
&&-\left( 2+\frac{n(n_{1}-1)(n_{2}-1)-2n_{1}n_{2}}{n\left( n-1\right)}
\right) ^{2}.
\end{eqnarray*}
Then, the proof is completed through some tedious calculations.

\subsection*{A.2. Proof of Theorem \ref{thm.EVm} }
By using the mean of the total number of runs, $R$, which was obtained by Marshall
\cite{marshall-70},
we have
$E\left(R\right)= 1+\frac{{\rm 2}n_{{\rm 1}}n_{{\rm 2}}}{n}$.
Also,
$E\left(R\right)=E\left(R_{{\rm M}}\right)+E\left(R_m\right)$.
So, by using the first part of Theorem
\ref{thm.EVM}, we have
\[E\left(R_m\right)=\frac{n_1n_2}{n-1}.\]

For the second part of the theorem,   similar to Theorem \ref{thm.EVM} and using the lemmas \ref{lem.Rm1} and \ref{lem.Rm2},
we have
\begin{eqnarray*}
Var\left(R_M\right)&=&\frac{n_{{\rm 1}}\left(n_{{\rm 1}}{\rm -}{\rm 1}\right)\left(n_{{\rm 2}}{\rm -}{\rm 2}\right)\left(n_{{\rm 2}}{\rm -}{\rm 1}\right)}{{\left(n{\rm -}{\rm 2}\right)}^{{\rm 2}}\left(n{\rm -}{\rm 3}\right)}{\rm \times }\frac{n_1\left(n_1-1\right)}{n\left(n-1\right)}\\
&&+\frac{n_{{\rm 2}}\left(n_{{\rm 2}}{\rm -}{\rm 1}\right)\left(n_{{\rm 1}}{\rm -}{\rm 2}\right)\left(n_{{\rm 1}}{\rm -}{\rm 1}\right)}{{\left(n{\rm -}{\rm 2}\right)}^{{\rm 2}}\left(n{\rm -}{\rm 3}\right)}{\rm \times \ }\frac{n_{{\rm 2}}\left(n_{{\rm 2}}{\rm -}{\rm 1}\right)}{n\left(n{\rm -}{\rm 1}\right)}
\\
&&{\rm +\ }\frac{{\left(n_{{\rm 2}}{\rm -}{\rm 1}\right)}^{{\rm 2}}{\left(n_{{\rm 1}}{\rm -}{\rm 1}\right)}^{{\rm 2}}}{{\left(n{\rm -}{\rm 2}\right)}^{{\rm 2}}\left(n{\rm -}{\rm 3}\right)}{\rm \times }\frac{{\rm 2}n_{{\rm 1}}n_{{\rm 2}}}{n\left(n{\rm -}{\rm 1}\right)}\\
&&+{\left({\rm 2+}\frac{{\rm (}n_{{\rm 1}}{\rm -}{\rm 1)}\left(n_{{\rm 2}}{\rm -}{\rm 2}\right)}{n{\rm -}{\rm 2}}\right)}^2\times \frac{n_1\left(n_1-1\right)}{n\left(n-1\right)}\\
&&+{\left({\rm 2+}\frac{\left(n_{{\rm 1}}{\rm -}{\rm 2}\right)\left(n_{{\rm 2}}{\rm -}{\rm 1}\right)}{n{\rm -}{\rm 2}}\right)}^2\times \frac{n_{{\rm 2}}{\rm (}n_{{\rm 2}}{\rm -}{\rm 1)}}{n\left(n{\rm -}{\rm 1}\right)}\\
&&+{\left({\rm 1+}\frac{\left(n_{{\rm 1}}{\rm -}{\rm 1}\right)\left(n_{{\rm 2}}{\rm -}{\rm 1}\right)}{n{\rm -}{\rm 2}}\right)}^2\times \frac{{\rm 2}n_{{\rm 1}}n_{{\rm 2}}}{n\left(n{\rm -}{\rm 1}\right)}-{\left(\frac{n_1n_2}{n-1}\right)}^2.
\end{eqnarray*}
The proof is completed through some tedious calculations.



\begin{thebibliography}{10}

\bibitem{asano-65}
C.~Asano.
\newblock Runs test for a circular distribution and a table of probabilities.
\newblock {\em Annals of the Institute of Statistical Mathematics},
  17(1):331--346, 1965.

\bibitem{ba-ko-02}
N.~Balakrishnan and M.~V. Koutras.
\newblock {\em Runs and Scans with Applications}.
\newblock Wiley, New York, 2002.

\bibitem{gi-ch-03}
J.~D. Gibbons and S.~Chakraborti.
\newblock {\em Nonparametric Statistical Inference}.
\newblock Marcel Dekker Inc., New York, 2003.

\bibitem{go-go-92}
A.~P. Godbole and M.~C. Gornowicz.
\newblock Exact and approximate runs distributions.
\newblock {\em Communications in Statistics-Theory and Methods},
  21(8):2151--2167, 1992.

\bibitem{govindarajulu-07}
Z.~Govindarajulu.
\newblock {\em Nonparametric Inference}.
\newblock World Scientific Publishing Co. Pte. Ltd, Singapore, 2007.

\bibitem{ho-cr-78}
R.~V. Hogg and A.~T. Craig.
\newblock {\em Introduction to Mathematical Statistic}.
\newblock Macmillan Publishing Co., New York, 1978.

\bibitem{lou-96}
W.~Y.~W. Lou.
\newblock On runs and longest run tests: a method of finite {M}arkov chain
  imbedding.
\newblock {\em Journal of the American Statistical Association},
  91(436):1595--1601, 1996.

\bibitem{marshall-70}
C.~W. Marshall.
\newblock A simple derivation of the mean and variance of the number of runs in
  an ordered sample.
\newblock {\em The American Statistician}, 24(4):27--28, 1970.

\bibitem{mood-40}
A.~M. Mood.
\newblock The distribution theory of runs.
\newblock {\em The Annals of Mathematical Statistics}, 11(4):367--392, 1940.

\bibitem{ra-wo-79}
R.~H. Randles and D.~A. Wolfe.
\newblock {\em Introduction to the Theory of Nonparametric Statistics}.
\newblock John Wiley and Sons, New York, 1979.

\bibitem{schwager-83}
S.~J. Schwager.
\newblock Run probabilities in sequences of {M}arkov-dependent trials.
\newblock {\em Journal of the American Statistical Association},
  78(381):168--175, 1983.

\bibitem{wa-wo-40}
A.~Wald and J.~Wolfowitz.
\newblock On a test whether two samples are from the same population.
\newblock {\em The Annals of Mathematical Statistics}, 11(2):147--162, 1940.

\bibitem{wolfowitz-44}
J.~Wolfowitz.
\newblock Asymptotic distribution of runs up and down.
\newblock {\em The Annals of Mathematical Statistics}, 15(2):163--172, 1944.

\end{thebibliography}

\end{document}